\documentclass[showpacs,preprintnumbers,amsmath,amssymb]{revtex4}
\newtheorem{theorem}{Theorem}
\newtheorem{prop}{Proposition}

\newtheorem{corollary}{Corollary}

\newtheorem{example}{Example}
\newcommand{\BlackBox}{\rule{1.5ex}{1.5ex}}  
\newenvironment{proof}{\par\noindent{\bf Proof:\
}}{\hfill\BlackBox\\[2mm]}

\usepackage{graphicx}
\usepackage{dcolumn}
\usepackage{bm}

\begin{document}
\title{
Uncertainty Relations for Joint Localizability and 
Joint Measurability in Finite-Dimensional Systems
}

\author{Takayuki Miyadera}
\affiliation{%
Research Center for Information Security, \\
National Institute of Advanced Industrial Science and Technology \\
1-1-1 Umezono, Tsukuba, 
Ibaraki 305-8561 Japan. 
\\
(E-mail: miyadera-takayuki@aist.go.jp)
}%


\date{Revised: 27 June, 2011}

\begin{abstract}
Two quantities quantifying uncertainty relations are examined. 
In J.Math.Phys. {\bf 48}, 082103 (2007), 
Busch and Pearson 
investigated the limitation on joint localizability and joint measurement 
of position and momentum
by introducing overall width and error bar width. 
In this paper, we show a simple relationship between 
these quantities for finite-dimensional systems. 
Our result indicates that if there is a bound on joint localizability, 
it is possible to obtain a similar bound on joint measurability. 
For finite-dimensional systems, uncertainty relations for a pair of general projection-valued measures are obtained as by-products.
\end{abstract}
\pacs{03.65.Ta}
\maketitle
\section{Introduction}
The uncertainty relation 
is one of the 
most fundamental features of quantum theory. 
This relation is often represented symbolically by an 
inequality
$\delta A \cdot \delta B \geq \mbox{`noncommutativity'}
$. 
Despite its simplicity, 
this inequality
has two distinct interpretations. One is related to the joint 
localizability of states with respect to observables $A$ and $B$. 
This property has been well investigated thus far. 
Depending on the measures for characterizing  
the (un)sharpness of distributions, 
their quantitative bounds can be represented by several inequalities. 
The most prominent one is Robertson's inequality \cite{Robertson}
employing standard deviation as the measure. 
The less prominent but important ones include 
the entropic uncertainty relation \cite{Deutsch,MU,KP} and the Landau-Pollak uncertainty 
relation \cite{MU,MiyaLP}.  
Another interpretation is related to the joint measurability of 
observables $A$ and $B$. 
That is, there is no measurement process that achieves joint measurement 
of noncommutative observables. 
Although this property
dates back to Heisenberg's paper 
\cite{Heisenberg} and 
has been utilized often in heuristic arguments, 
only a few rigorous quantitative representations 
have been obtained thus far. 
While investigating joint measurement of the position and momentum of a quantum mechanical particle, Appleby \cite{Appleby} introduced error operators and disturbance operators and derived simple inequalities between them.
Ozawa \cite{OzawaUncertain} treated a pair of general self-adjoint 
observables and 
considered a trade-off relation between 
his error operator and disturbance operator that 
have an interpretation in the context of his 
extended notion of 
joint measurement. Werner \cite{Werner} formulated the problem 
operationally and derived an inequality between position and 
momentum. Busch and Pearson \cite{Buscherrorbar} introduced the concept
 of error bar width 
that has a clear meaning operationally, 
and discussed its trade-off relation between position and momentum. 
Miyadera and Imai \cite{MiyaHeisen},  while employing
 simple distances between a pair of 
probability distributions, derived a trade-off inequality 
between a pair of general discrete positive-operator-valued measures (POVMs). 
Busch, Heinonen, and Lahti submitted a review on this topic \cite{BuschUncertainty}. 
\par
In this paper, we study the relationship between joint localizability and 
joint measurability. We employ the formulation introduced 
by Busch and Pearson \cite{BuschUncertainty}. 
In their study of a quantum mechanical particle, 
quantities called overall width 
and error bar width played important roles. 
The overall width of a probability distribution represents 
 the width of its localization. For any state, the overall widths of its
position representation and momentum representation 
cannot be very small simultaneously. 
The error bar width measures the quality of approximate joint measurement. 
Busch and Pearson considered a general observable $M$ to jointly 
measure position $Q$ and momentum $P$ approximately. 
The error bar width of $M$ with respect to $Q$ quantifies how well 
$M$ approximates $Q$. It represents the width of the region into which an outcome 
of $M$ can fall when a state is prepared sharply with respect to $Q$. 
It was shown that 
the error bar width of $M$ with respect to $Q$ and 
that of $M$ with respect to $P$ cannot be very small 
simultaneously. 
Interestingly, 
while characterizing 
distinct properties, the overall width and error bar width
satisfy inequalities whose forms 
are identical. 
We examine this correspondence for a general pair of 
projection-valued measures (PVMs)
in finite-dimensional Hilbert spaces.  
In the next section, we present a brief review of the formulation introduced 
by Busch and Pearson. In section \ref{mainsection}, we describe our main 
result, which shows a simple relationship between overall width and error bar 
width. Our result implies that if there is a bound on joint localizability, 
it is possible to obtain a similar bound on joint measurability. 
We derive uncertainty relations for joint localizability and joint 
measurability as by-products. 
In section \ref{others}, we discuss other measures for joint measurement. 
We conclude the paper with a short discussion. 
\section{Formulation}
\subsection{Approximate joint localization}
Position and momentum cannot be strictly localized simultaneously 
in a quantum state.
That is, if a state $\rho$ is strictly bounded in 
a position representation, its support in a momentum representation 
cannot be bounded. From an operational viewpoint, this suggests 
the following. 
Suppose we measure position and 
momentum individually. 
If a state shows a strictly localized distribution of the outcomes for 
the position measurement, the same state cannot 
have a strictly localized distribution of the outcomes for the momentum 
measurement. 
This feature of ``no joint localizability" 
holds for a pair of general observables. 
Let us consider a pair of PVMs $A$ on 
a measurable space $\Omega_A$ and $B$ on a measurable space $\Omega_B$. 
If measurable subsets 
$\Delta_A \subset \Omega_A$ and 
$\Delta_B \subset \Omega_B$ satisfy $\Vert A(\Delta_A)B(\Delta_B)\Vert \neq 1$, 
any state $\rho$ cannot satisfy $\mbox{tr}(\rho A(\Delta_A))
=\mbox{tr}(\rho B(\Delta_B))=1$. This result can be easily derived 
by the Landau-Pollak uncertainty relation \cite{Deutsch,MU,KP,MiyaLP}, 
\begin{eqnarray*}
\mbox{tr}(\rho A(\Delta_A))+\mbox{tr}(\rho B(\Delta_B))
\leq 1 +\Vert A(\Delta_A)B(\Delta_B)\Vert. 
\end{eqnarray*}
This impossibility theorem has its quantitative counterpart. 
In \cite{Buscherrorbar}, Busch and Pearson introduced a 
quantity called overall width 
that characterizes the width of localization. 
While they exclusively treated canonical position 
and momentum, their formulation is broad enough to cover 
a wide class of observables. 
Let us consider a probability measure ${\bf p}$ on 
a metric space $(\Omega,d)$, where $\Omega$ is a set and 
$d$ is a metric (distance) defined on it \cite{metric}. 
The overall width (at confidence level $1-\epsilon$) of 
the probability ${\bf p}$ is 
defined for $\epsilon \in [0,1)$ by 
\begin{eqnarray*}
W_{\epsilon}({\bf p}):=
\mbox{inf}\{w>0| 
\exists x\in \Omega: 
{\bf p}(O_d(x,w))\geq 1-\epsilon\}, 
\end{eqnarray*}
where a ball $O_d(x,L)$ for $x\in \Omega$ and 
$L\geq 0$ is defined 
by $O_d(x,L):=\{y\in \Omega| d(x,y)\leq \frac{L}{2}\}$. 
Let us consider a quantum system described by a 
Hilbert space ${\cal H}$. The set of all density operators 
is denoted by ${\cal S}({\cal H})$, and the set of all
bounded operators is denoted by ${\bf B}({\cal H})$.  
For a POVM $F$ on 
a measurable space $\Omega$,  
and for a state $\rho\in {\cal S}({\cal H})$, 
$\rho^F$ denotes the probability measure defined by 
$\rho^F(X):=\mbox{tr}(\rho F(X))$. 
In general, 
for a pair of noncommutative  
PVMs $A$ and $B$, 
$\rho^A$ and $\rho^B$ cannot be too sharp simultaneously. 
This property is expressed by a quantitative trade-off inequality 
between $W_{\epsilon_1}(\rho^A)$ and $W_{\epsilon_2}(\rho^B)$ 
for small $\epsilon_1, \epsilon_2 \in [0,1]$. 
In fact, Busch, Heinonen and Lahti \cite{BuschUncertainty} proved that 
between position $Q$ and momentum $P$ of a particle, 
it holds that
\begin{eqnarray}
W_{\epsilon_1}(\rho^Q)\cdot W_{\epsilon_2}(\rho^P)
\geq 2\pi \hbar \cdot (1-\epsilon_1 -\epsilon_2)^2, 
\label{QPlocal}
\end{eqnarray}
for any state $\rho$ and $\epsilon_1, \epsilon_2>0$ with $\epsilon_1+\epsilon_2\leq 1$, 
where $\Omega_P=\Omega_Q={\bf R}$ is equipped with Euclidean distance 
$d(x,y)=|x-y|$. 
\subsection{Approximate joint measurement}
Position and momentum cannot be measured simultaneously, 
which was first identified by Heisenberg using a Gedankenexperiment. 
Busch and Pearson formulated the limitation for 
joint measurement of position and momentum in terms of 
an error bar. 
Their formulation is based on the following 
operationally relevant idea. 
To estimate the quality of a measurement 
device, one may prepare a state that is well localized 
with respect to an ideal observable and then measure 
an approximate observable to detect the magnitude of error. 
The precise definition is as follows. 
Suppose $E$ is a PVM on a metric space $(\Omega,d)$.  
Let $E_1$ be a POVM 
on $(\Omega,d)$. For each $\epsilon\in [0,1]$, $\delta> 0$,  
the error of $E_1$ relative to $E$ is 
defined by 
\begin{eqnarray*}
{\cal W}_{\epsilon,\delta}(E_1,E)
:= \mbox{inf}
\{w>0| \forall x\in \Omega, \forall \rho \in {\cal S}({\cal H}): 
\rho^E(O_d(x,\delta))=1 \Rightarrow 
\rho^{E_1}(O_d(x,w))\geq 1-\epsilon\}.
\end{eqnarray*}
The error is an increasing function of $\delta$. 
Busch and Pearson defined the error bar width 
of $E_1$ relative to $E$ as 
\begin{eqnarray*}
{\cal W}_{\epsilon}(E_1,E)
:=\mbox{inf}_{\delta} {\cal W}_{\epsilon,\delta}
(E_1,E)
=\lim_{\delta \to 0}
{\cal W}_{\epsilon,\delta}(E_1,E).
\end{eqnarray*}
\par
To discuss an approximate joint measurement of 
a pair of PVMs $A$ on $\Omega_A$ and $B$ on $\Omega_B$, 
we consider a POVM $M$ defined on 
$\Omega_A\times \Omega_B$. 
$M$ naturally defines a POVM $M_1$ (resp. $M_2$) on $\Omega_A$ (resp. $\Omega_B$) by 
$M_1(\Delta):=M(\Delta \times \Omega_B)$ 
(resp. $M_2(\Delta'):=M(\Omega_A\times \Delta')$ ). 
The quality of approximation is characterized by 
${\cal W}_{\epsilon_1}(M_1,A)$ and 
${\cal W}_{\epsilon_2}(M_2,B)$. 
When both values vanish simultaneously for some 
POVM $M$, $A$ and $B$ are said to be jointly measurable. 
Busch and Pearson proved that for the position and momentum 
of a quantum particle, 
\begin{eqnarray}
{\cal W}_{\epsilon_1}(M_1,Q)\cdot 
{\cal W}_{\epsilon_2}(M_2,P) \geq 2\pi \hbar 
(1-\epsilon_1-\epsilon_2)^2 \label{QPerrorbar}
\end{eqnarray}
holds for $\epsilon_1,\epsilon_2>0$ with $\epsilon_1+\epsilon_2\leq 1$. 
\par
For a finite $\Omega$, the error bar width can be expressed without 
introducing $\delta$. In fact, it holds that 
\begin{eqnarray}
{\cal W}_{\epsilon}(E_1,E)
= \mbox{inf}
\{w>0| \forall x\in \Omega, \forall \rho \in {\cal S}({\cal H}): 
\rho^E(\{x\})=1 \Rightarrow 
\rho^{E_1}(O_d(x,w))\geq 1-\epsilon\}.
\label{defdiscrete}
\end{eqnarray}
Because we treat 
observables with finite outcomes in this paper, 
this expression is employed. 
In the following, a POVM on a finite set $\Omega$ is 
written explicitly for each point in $\Omega$. 
That is, by a POVM  $\{A_x\}_{x\in \Omega}$, we mean 
that for each $x\in \Omega$, 
$A_x\geq 0$ is satisfied, and $\sum_{x\in \Omega} A_x={\bf 1}$ holds. 
\section{Localizability and measurability in a finite-dimensional system}
\label{mainsection}
As described above, 
Busch and Pearson showed a trade-off inequality between 
the error bar width 
${\cal W}_{\epsilon_1}(M_1,Q)$ and ${\cal W}_{\epsilon_2}(M_2,P)$.
It should be noted that the form of the bound for joint measurement 
(\ref{QPerrorbar}) is identical with 
that for joint localization (\ref{QPlocal}).
In the following, we show that this type of agreement 
between joint localizability and joint measurement generally  
exists for finite-dimensional Hilbert spaces. 
The following theorem shows a simple relationship between 
the overall width and error bar width. 
On the basis of this theorem, one can see that if there exists a bound on 
joint localizability, a similar bound holds for 
joint measurability. 
\begin{theorem}\label{th1}
Let ${\cal H}$ be an $N$-dimensional Hilbert space. 
Suppose that we have a pair of PVMs $A=\{A_x\}_{x\in \Omega_A}$ 
and $B=\{B_y\}_{y \in \Omega_B}$, where $(\Omega_A,d_A)$ 
and $(\Omega_B,d_B)$ are finite metric spaces.  
We consider a POVM $M=\{M_{xy}\}_{(x,y)\in \Omega_A\times \Omega_B}$ 
to perform an approximate joint measurement of $A$ and $B$.
$M_1$ and $M_2$ denote marginal POVMs obtained from $M$.  
\par
For any $\epsilon_1, \epsilon_2 \in [0,1]$ with $\epsilon_1+\epsilon_2\leq 1$, 
there exists $\rho\in {\cal S}({\cal H})$ satisfying
\begin{eqnarray*}
{\cal W}_{\epsilon_1}(M_1,A)&\geq& W_{\epsilon_1+\epsilon_2}
(\rho^{A}) \\
{\cal W}_{\epsilon_2}(M_2,B)&\geq& W_{\epsilon_1+\epsilon_2}
(\rho^{B}).
\end{eqnarray*}
\end{theorem}
\begin{proof}
Because  
$\Omega_A$ and $\Omega_B$ are finite sets ($|\Omega_A|, |\Omega_B|\leq N$), 
we employ (\ref{defdiscrete}) to discuss the 
error bar width. 
For each $x\in \Omega_A$, a state $\rho_x$ is chosen 
so as to satisfy $\rho_x^A(\{x\})=1$. 
(Note that $\rho_x$ is not uniquely determined in general, because 
$A_x$ may be degenerate.)
Let $\{p_A(x)\}_{x\in \Omega_A}$ be a probability 
distribution on $\Omega_A$. 
Suppose that we prepare $\rho_x$ with 
probability $p_A(x)$ and then measure $M$. 
We denote by $p_{M|A}(x',y'|x)$ 
the conditional probability to obtain $(x',y') \in \Omega_A 
\times \Omega_B$ when we prepared $\rho_x$. 
According to the definition of ${\cal W}_{\epsilon_1}(M_1,A)$, 
for any $w_1 \geq {\cal W}_{\epsilon_1}(M_1,A)$ and for each $x\in \Omega_A$, 
it holds that 
\begin{eqnarray*}
\sum_{x'\in O_{d_A}(x,w_1)}\ 
\sum_{y'\in \Omega_B}p_{M|A}(x',y'|x)\geq 1-\epsilon_1. 
\end{eqnarray*}
This inequality implies 
\begin{eqnarray}
\sum_{x\in \Omega_A}\ 
\sum_{x'\in O_{d_A}(x,w_1)}
\sum_{y'\in \Omega_B}p_{M|A}(x',y'|x)p_A(x)\geq 1-\epsilon_1. 
\label{eqpo}
\end{eqnarray}
Let us denote by $p_{A|M}(x|x',y')$ a posterior probability 
representing that the prepared state is $\rho_x$ 
when 
$(x',y')
\in \Omega_A\times \Omega_B$ is obtained by the measurement of $M$. 
According to the Bayes rule,  
$p_{M|A}(x',y'|x)p_A(x)=p^{A}_M(x',y')p_{A|M}(x|x',y')$ holds, 
where $p^{A}_M(x',y')$ is the probability to obtain 
outcome $(x',y')\in \Omega_A\times \Omega_B$. 
Thanks to the symmetry of the metric, 
the above inequality (\ref{eqpo}) can be rewritten as 
\begin{eqnarray}
\sum_{(x',y')\in \Omega_A\times \Omega_B}\ 
\sum_{x\in O_{d_A}(x',w_1)}
p^A_M(x',y')p_{A|M}(x|x',y')\geq 1-\epsilon_1. 
\label{Asiki}
\end{eqnarray} 
\par
Similarly, for each $y\in \Omega_B$, a state $\sigma_y$ is chosen 
so as to satisfy $\sigma_y^B(\{y\})=1$. 
Let $\{p_B(y)\}_{y\in \Omega_B}$ be a probability 
distribution on $\Omega_B$. 
Suppose that we prepare $\sigma_y$ with 
probability $p_B(y)$ and then measure $M$. 
We denote by $p_{M|B}(x',y'|y)$ 
the conditional probability to obtain $(x',y') \in \Omega_A 
\times \Omega_B$ when we prepared $\sigma_y$. 
For any $w_2\geq {\cal W}_{\epsilon_2}(B,M_2)$, it holds that
\begin{eqnarray*}
\sum_{x'\in \Omega_A}\ 
\sum_{y'\in O_{d_B}(y,w_2)}
p_{M|B}(x',y'|y)\geq 1-\epsilon_2. 
\end{eqnarray*}
We thus obtain 
\begin{eqnarray}
\sum_{(x',y')\in \Omega_A\times \Omega_B}\ 
\sum_{y\in O_{d_B}(y',w_2)}
p^B_M(x',y')p_{B|M}(y|x',y')\geq 1-\epsilon_2,   
\label{Bsiki}
\end{eqnarray} 
where $p^{B}_M(x',y')$ is the probability to obtain 
$(x',y')\in \Omega_A\times \Omega_B$ by measuring $M$ 
in this protocol, 
and $p_{B|M}(y|x',y')$ denotes a posterior probability representing 
that the prepared state is $\sigma_y$ 
when $(x',y')
\in \Omega_A\times \Omega_B$ is obtained by measuring $M$. 
\par
While the above two probabilities $p^A_M(x',y')$ and 
$p^B_M(x',y')$ are different in general, 
they coincide with each other if $\sum_{x\in \Omega_A}
p_A(x)\rho_x =\sum_{y \in \Omega_B}p_B(z)\sigma_z$ is satisfied. 
We construct such a pair of ensembles below. 
\par
We employ the remote state preparation technique. 
${\cal H}={\bf C}^N$ has an orthonormalized basis 
$\{{\bf e}_i\}$ $(i=1,2,\ldots,N)$. 
Let us consider an auxiliary $N$-dimensional Hilbert space 
${\cal H}_{aux}$ 
and the composite system ${\cal K}:={\cal H}_{aux} \otimes {\cal H}$. 
We denote by $\{{\bf a}_i\}_{i=1,2,\ldots,N}$ an orthonormalized basis of ${\cal H}_{aux}$. 
A normalized vector $\Phi_0:=\frac{1}{\sqrt{N}}
\sum_i {\bf a}_i \otimes {\bf e}_i$ is a maximally entangled state 
over the composite system. 
We define an antiunitary operator $J :{\cal K} \to {\cal K}$ by
\begin{eqnarray*}
J(\sum_i \sum_j c_{ij}{\bf a}_i \otimes {\bf e}_j)
=\sum_i \sum_j \overline{c_{ij}}{\bf a}_j\otimes {\bf e}_i, 
\end{eqnarray*}
which satisfies $J^2={\bf 1}$. 
It is easy to see that 
for an arbitrary operator $X\in {\bf B}({\cal H})$, 
$({\bf 1}\otimes X)\Phi_0
=J({\bf 1}\otimes X^*)J\Phi_0 $ holds.  
Moreover,  $\tilde{X}=\sum_{i,j}({\bf e}_j,X{\bf e}_i)|{\bf a}_i\rangle
\langle {\bf a}_j| \in {\bf B}({\cal H}_{aux})$ satisfies 
$J({\bf 1}\otimes X^*)J =\tilde{X}\otimes {\bf 1}$. 
This correspondence defines a map $\gamma:
{\bf B}({\cal H}) \to {\bf B}({\cal H}_{aux})$ 
by $\gamma(X)=\tilde{X}$. 
This map $\gamma$ is a bijective linear map 
from ${\bf B}({\cal H})$ to ${\bf B}({\cal H}_{aux})$. 
It maps the identity operator ${\bf 1}$ on ${\cal H}$ to 
the identity operator ${\bf 1}$ on ${\cal H}_{aux}$. 
In addition, it is a positive map. That is,   
$\gamma(X)\geq 0$ holds for $X\geq 0$. 
Therefore we can define its dual $\gamma^*$ 
as a map from ${\cal S}({\cal H}_{aux})$ 
to ${\cal S}({\cal H})$ satisfying 
$\mbox{tr}(\rho \gamma(X))=\mbox{tr}(\gamma^*(\rho)X)$ 
for all $X\in {\bf B}({\cal H})$. 
In addition, $\gamma(X)$ becomes 
a projection for projection $X$.  
\par
For PVM $A=\{A_x\}$, we 
define $p_A(x):=(\Phi_0, ({\bf 1}\otimes A_x)\Phi_0)$ 
and 
\begin{eqnarray*}
\rho_x:=\frac{\mbox{tr}_{{\cal H}_{aux}}(({\bf 1}\otimes A_x)|\Phi_0
\rangle \langle \Phi_0|({\bf 1}\otimes A_x))}
{p_A(x)},
\end{eqnarray*}
where $\mbox{tr}_{{\cal H}_{aux}}$ represents 
a partial trace with respect to ${\cal H}_{aux}$. 
It is easy to see that this state satisfies $\rho_x^A(\{x\})=1$. 
In addition, $\sum_{x}p_A(x)\rho_x=\frac{{\bf 1}}{N}$ holds. 
In fact, as $\{\gamma(A_x)\}_{x\in \Omega_A}$ becomes a PVM on ${\cal H}_{aux}$, 
we obtain for any operator $X$ on ${\cal H}$, 
\begin{eqnarray*}
\sum_{x\in \Omega_A}p(x)\mbox{tr}(\rho_x X)
=\sum_{x\in \Omega_A}( \Phi_0,
(\gamma(A_x)\otimes X)\Phi_0)
=( \Phi_0, ({\bf 1}\otimes X) \Phi_0)
=\mbox{tr}\left(\frac{{\bf 1}}{N}X
\right).
\end{eqnarray*}
Similarly, for each $y\in \Omega_B$ we put 
$p_B(y)=(\Phi_0, ({\bf 1}\otimes B_y)\Phi_0)$ and 
\begin{eqnarray*}
\sigma_y:=\frac{\mbox{tr}_{{\cal H}_{aux}}(({\bf 1}\otimes B_y)|\Phi_0
\rangle \langle \Phi_0|({\bf 1}\otimes B_y))}
{p_B(y)},
\end{eqnarray*}
which satisfies $\sum_{y\in \Omega_B}p_B(y)\sigma_y=\frac{\bf 1}{N}$. 
\par
Because the ensembles constructed above satisfy 
$\sum_{x\in \Omega_A}p_A(x)\rho_x=
\sum_{y\in \Omega_B}p_B(y)\sigma_z$,  
$p^A_M(x,y)=p^B_M(x,y)=:p_M(x,y)$ holds for each $(x,y)\in \Omega_A
\times \Omega_B$. 
Therefore, (\ref{Asiki}) and (\ref{Bsiki}) 
can be combined to show 
\begin{eqnarray*}
\sum_{(x',y')\in \Omega_A\times \Omega_B}
p_M(x',y')
\left( \sum_{x\in O_{d_A}(x',w_1)}
p_{A|M}(x|x',y')+
\sum_{y\in O_{d_B}(y',w_2)}
p_{B|M}(y|x',y')\right)
\geq 2-\epsilon_1-\epsilon_2. 
\end{eqnarray*}
This inequality implies that there exists 
$(x',y')\in \Omega_A\times \Omega_B$ satisfying $p_M(x',y')\neq 0$ and 
\begin{eqnarray}
\sum_{x\in O_{d_A}(x',w_1)}
p_{A|M}(x|x',y')+
\sum_{y\in O_{d_B}(y',w_2)}
p_{B|M}(y|x',y')
\geq 2-\epsilon_1-\epsilon_2. 
\label{nanda}
\end{eqnarray}
We rewrite the left-hand side of the above inequality below. 
Let us note that the joint probability $p_{M|A}(x',y'|x)p_A(x)$ 
can be expressed as 
\begin{eqnarray*}
p_{M|A}(x',y'|x)p_A(x)=p_A(x)\mbox{tr}(\rho_x M_{x'y'})
=(\Phi_0,(\gamma(A_x)\otimes M_{x'y'})\Phi_0).
\end{eqnarray*}
Because $p_M(x',y')=(\Phi_0, ({\bf 1}\otimes M_{x'y'}) \Phi_0)$ holds, 
using the Bayes rule we obtain 
\begin{eqnarray*}
p_{A|M}(x|x',y')=\frac{(\Phi_0,(\gamma(A_x)\otimes M_{x'y'})\Phi_0)}
{(\Phi_0,({\bf 1}\otimes M_{x'y'})\Phi_0)}.
\end{eqnarray*}
Defining a state $\Theta_{x'y'}$ of an auxiliary system ${\cal H}_{aux}$ 
by 
\begin{eqnarray*}
\Theta_{x'y'}:=\frac{\mbox{tr}_{\cal H}\left(
({\bf 1}\otimes M_{x'y'}^{1/2})|\Phi_0\rangle 
\langle \Phi_0|({\bf 1}\otimes M_{x'y'}^{1/2})
\right)}
{(\Phi_0,({\bf 1}\otimes M_{x'y'})\Phi_0)
}, 
\end{eqnarray*}
we obtain
$p_{A|M}(x|x',y')=\mbox{tr}(\Theta_{x'y'}\gamma(A_x))$.
Similarly, we obtain 
$p_{B|M}(y|x',y')=\mbox{tr}(\Theta_{x'y'}\gamma(B_y))$.
Thus (\ref{nanda}) can be written as 
\begin{eqnarray*}
\mbox{tr}\left(\Theta_{x'y'}\sum_{x
\in O_{d_A}(x',w_1)}
\gamma(A_x)
\right)
+
\mbox{tr}\left(\Theta_{x'y'}\sum_{y\in 
O_{d_B}(y',w_2)}
\gamma(B_y)
\right)
\geq 2-\epsilon_1-\epsilon_2.
\end{eqnarray*}
This can be rewritten as
\begin{eqnarray}
\mbox{tr}\left(\gamma^*(\Theta_{x'y'})
\sum_{x\in 
O_{d_A}(x',w_1)}
A_x
\right)
+
\mbox{tr}\left(\gamma^*(\Theta_{x'y'})
\sum_{y\in 
O_{d_B}(y',w_2)}
B_y
\right)
\geq 2-\epsilon_1-\epsilon_2,
\label{key}
\end{eqnarray}
where $\gamma^*:{\cal S}({\cal H_A})\to {\cal S}({\cal H})$ 
is the dual map of $\gamma$. 
Because each term of the left-hand side of (\ref{key}) 
is smaller than $1$, 
 we obtain an important observation:
for all $\epsilon_1,\epsilon_2\in [0,1]$ with $\epsilon_1+\epsilon_2\leq 1$ it holds 
that 
\begin{eqnarray*}
w_1 &\geq& W_{\epsilon_1+\epsilon_2}(\gamma^*(\Theta_{x'y'})^{A})
\\
w_2 &\geq& W_{\epsilon_1+\epsilon_2}(\gamma^*(\Theta_{x'y'})^{B}). 
\end{eqnarray*}
Because $w_1$ and $w_2$ are arbitrary numbers satisfying 
$w_1\geq {\cal W}_{\epsilon_1}(M_1,A)$ and $w_2\geq {\cal W}_{\epsilon_2}
(M_2,B)$, we obtain 
\begin{eqnarray}
{\cal W}_{\epsilon_1}(M_1,A)&\geq& W_{\epsilon_1+\epsilon_2}
(\gamma^*(\Theta_{x'y'})^{A}) \label{w1w1}\\
{\cal W}_{\epsilon_2}(M_2,B)&\geq& W_{\epsilon_1+\epsilon_2}
(\gamma^*(\Theta_{x'y'})^{B}) \label{w2w2}.
\end{eqnarray}
\end{proof}
Thanks to this theorem, we can estimate the limitation on  
joint measurement by using a limitation on joint localization. 
That is, if we have a bound on joint localization in the state-independent 
form, we can obtain a corresponding bound on joint measurement. 
\par
While the above theorem shows a directly 
applicable relationship, it is often possible to 
obtain a better bound by employing (\ref{key}) 
as a starting point. 
We obtain the following theorem. 
\begin{theorem}
Let ${\cal H}$ be an $N$-dimensional Hilbert space. 
Suppose that we have a pair of PVMs $A=\{A_x\}_{x\in \Omega_A}$ 
and $B=\{B_y\}_{y \in \Omega_B}$, where $(\Omega_A,d_A)$ 
and $(\Omega_B,d_B)$ are metric spaces.  
We consider a POVM $M=\{M_{xy}\}_{(x,y)\in \Omega_A\times \Omega_B}$ 
to perform an approximate joint measurement of $A$ and $B$.
$M_1$ and $M_2$ denote marginal POVMs obtained from $M$.  
The following statements hold. 
\\
(i)
For
$\epsilon_1, \epsilon_2\in [0,1]$ with $\epsilon_1 +\epsilon_2\leq 1$ and for 
any state $\rho$, the overall widths 
$W_{\epsilon_1}(\rho^A)$ and $W_{\epsilon_2}(\rho^B)$ 
satisfy
\begin{eqnarray*}
\max_{x\in\Omega_A} \left| O_{d_A}\left(x,W_{\epsilon_1}(\rho^A)
\right)
\right|
\cdot
\max_{y \in \Omega_B}\left|
O_{d_B}\left(y,W_{\epsilon_2}(\rho^B)\right)
\right|\geq
\frac{(1-\epsilon_1-\epsilon_2)^2}{ \max_{x\in \Omega_A,
y\in \Omega_B}\mbox{tr}( A_x B_y )}, 
\end{eqnarray*}
where $\left| X \right|$ means the counting measure (cardinality) of the set $X$. 
\\
(ii)
For
$\epsilon_1, \epsilon_2\in [0,1]$ with $\epsilon_1+\epsilon_2\leq 1$, 
the error bar widths ${\cal W}_{\epsilon_1}(M_1,A)$
and ${\cal W}_{\epsilon_2}(M_2,B)$ satisfy
\begin{eqnarray*}
\max_{x\in\Omega_A} \left| O_{d_A}\left(x,{\cal W}_{\epsilon_1}
(M_1,A)
\right)
\right|
\cdot
\max_{y \in \Omega_B}\left|
O_{d_B}\left(y,{\cal W}_{\epsilon_2}(M_2,B)\right)
\right|\geq
\frac{(1-\epsilon_1-\epsilon_2)^2}{ \max_{x\in \Omega_A,
y\in \Omega_B}\mbox{tr}( A_x B_y )}. 
\end{eqnarray*}
\end{theorem}
\begin{proof}
(i) 
We employ 
a strategy of the proof for (\ref{QPlocal}) in \cite{BuschUncertainty}. 
According to the Landau-Pollak uncertainty relation, 
for any state $\rho$ it holds that
\begin{eqnarray*}
\sum_{x'\in O_{d_A}(x,W_{\epsilon_1}(\rho))}
\mbox{tr}(\rho A_{x'})
+
\sum_{y'\in O_{d_B}(y,W_{\epsilon_2}(\rho))}
\mbox{tr}(\rho B_{y'})
\leq 1 +\left\Vert 
\sum_{x'\in O_{d_A}(x,W_{\epsilon_1}(\rho))} A_{x'}
\sum_{y'\in O_{d_B}(y,W_{\epsilon_2}(\rho))} B_{y'}
\right\Vert.
\end{eqnarray*}
Because $\Vert a\Vert \leq \mbox{tr}(a^*a)^{1/2}$ holds for 
an arbitrary operator $a$, 
we obtain 
\begin{eqnarray*}
\left\Vert 
\sum_{x'\in O_{d_A}(x,W_{\epsilon_1}(\rho))} A_{x'}
\sum_{y'\in O_{d_B}(y,W_{\epsilon_2}(\rho))} B_{y'}
\right\Vert
&\leq& 
\left(
\sum_{x'\in O_{d_A}\left(x,W_{\epsilon_1}(\rho^A)\right)}
\sum_{y'\in O_{d_B}\left(y,W_{\epsilon_2}(\rho^B)\right)}
\mbox{tr}\left(B_{y'} A_{x'}B_{y'}\right)
\right)^{1/2}
\\
&\leq &
\max_{x\in\Omega_A} \left| O_{d_A}\left(x,W_{\epsilon_1}(\rho^A)
\right)
\right|^{1/2}
\cdot
\max_{y \in \Omega_B}\left|
O_{d_B}\left(y,W_{\epsilon_2}(\rho^B)\right)
\right|^{1/2}
\\
&&
\cdot
\max_{x'\in \Omega_A,y'\in \Omega_B}
\mbox{tr}\left(A_{x'}B_{y'}
\right)^{1/2},
\end{eqnarray*}
where we used $B_{y_1}B_{y_2}=\delta_{y_1 y_2}B_{y_1}$. 
Because there exist $x\in \Omega_A$ and $y\in \Omega_B$ 
such that 
\begin{eqnarray*}
(1-\epsilon_1)+(1-\epsilon_2)
\leq 
\sum_{x'\in O_{d_A}(x,W_{\epsilon_1}(\rho))}
\mbox{tr}(\rho A_{x'})
+
\sum_{y'\in O_{d_B}(y,W_{\epsilon_2}(\rho))}
\mbox{tr}(\rho B_{y'})
\end{eqnarray*}
the claim (i) follows. 
\par
(ii)
Instead of applying Theorem \ref{th1} directly, 
we start with (\ref{key}) to obtain a better bound. 
Putting $w_1={\cal W}_{\epsilon_1}(M_1,A)$ and 
$w_2={\cal W}_{\epsilon_2}(M_2,B)$, we obtain 
\begin{eqnarray*}
\sum_{x\in 
O_{d_A}(x',{\cal W}_{\epsilon_1}(M_1,A))}
\mbox{tr}\left(\gamma^*(\Theta_{x'y'})
A_x
\right)
+
\sum_{y\in 
O_{d_B}(y',{\cal W}_{\epsilon_2}(M_2,B))}
\mbox{tr}\left(\gamma^*(\Theta_{x'y'})
B_y
\right)
\geq 2-\epsilon_1-\epsilon_2.
\end{eqnarray*}
Mimicking the argument in the proof of (i), we obtain the claim (ii).
\end{proof}
The above theorem can be applied to the following examples.
\begin{example}(A qubit)
\par
As the simplest example we consider a qubit described by a two-dimensional 
Hilbert space ${\cal H}={\bf C}^2$. We write a standard basis $|z\rangle$ 
for $z\in \{0,1\}$. Its conjugate basis is written as $|\overline{x}\rangle$, 
where $|\overline{0}\rangle:=\frac{1}{\sqrt{2}}(|0\rangle +|1\rangle)$ 
and $|\overline{1}\rangle :=\frac{1}{\sqrt{2}}(|0\rangle -|1\rangle)$. 
We consider a pair of PVMs $X=\{X_x\}_{x\in \Omega_X}
=\{|\overline{x}\rangle \langle \overline{x}|\}_{x\in \Omega_X}$ and 
$Z=\{Z_z\}_{z\in \Omega_Z}=\{|z\rangle \langle z|\}_{z\in \Omega_{Z}}$, 
both of which are defined on a metric space $\Omega_X=
\Omega_Z=\{0,1\}$. This space $\{0,1\}$ is equipped with 
a distance defined by $d(0,1)=d(1,0)=1$. The noncommutativity between 
$X$ and $Z$ is quantified as 
$\mbox{tr}(X_x Z_z)=\frac{1}{2}$ for each $x,z\in \{0,1\}$. 
Thus it holds for any $\epsilon_1,\epsilon_2\in [0,1]$ with 
$\epsilon_1 +\epsilon_2 \leq 1$ that 
\begin{eqnarray}
\max_{z\in \Omega_Z} \left|
O_d \left(z, W_{\epsilon_1}(\rho^Z)\right)
\right| \cdot 
\max_{x\in \Omega_X}\left| 
O_d \left(
x,W_{\epsilon_2}(\rho^X)
\right)
\right|
\geq 2(1-\epsilon_1-\epsilon_2)^2. 
\label{qq}
\end{eqnarray}
For $\epsilon_1=\epsilon_2=0$, it follows that  
\begin{eqnarray*}
\max_{z\in \Omega_Z} \left|
O_d \left(z, W_{0}(\rho^Z)\right)
\right| \cdot 
\max_{x\in \Omega_X}\left| 
O_d \left(
x,W_{0}(\rho^X)
\right)
\right|
\geq 2. 
\end{eqnarray*}
This inequality implies that 
if $\max_{x\in \Omega_X} \left|
O_d \left(x, W_{0}(\rho^X)\right)
\right| =1$ holds, $\max_{z\in \Omega_Z}\left| 
O_d \left(
z,W_{0}(\rho^Z)
\right)
\right|=2$ must be satisfied. 
That is, there is no state that simultaneously yields a sharp 
localization for $X$ and $Z$. 
In addition, (\ref{qq}) implies that 
$\epsilon_1+\epsilon_2 \geq 1-\frac{1}{\sqrt{2}}$ must be satisfied 
so that both $\max_{z\in \Omega_Z} \left|
O_d \left(z, W_{\epsilon_1}(\rho^Z)\right)
\right|=1$ and $\max_{x\in \Omega_X}\left| 
O_d \left(
x,W_{\epsilon_2}(\rho^X)
\right)
\right|=1$ hold. 
That is, it holds for any state $\rho$ that 
\begin{eqnarray*}
\max_{z\in \Omega_Z}\rho^Z(\{z\})
+\max_{x\in \Omega_X}\rho^X(\{x\})
\leq 1+\frac{1}{\sqrt{2}}.
\end{eqnarray*}
This inequality agrees with the Landau-Pollak uncertainty relation. 
\par
Similarly, it holds for any POVM $M$ on $\Omega_Z
\times \Omega_X$ that 
\begin{eqnarray*}
\max_{z\in \Omega_Z} \left|
O_d \left(z, {\cal W}_{\epsilon_1}(M_1,Z)\right)
\right| \cdot 
\max_{x\in \Omega_X}\left| 
O_d \left(
x,{\cal W}_{\epsilon_2}(M_2,X)
\right)
\right|
\geq 2(1-\epsilon_1-\epsilon_2)^2. 
\end{eqnarray*}
This inequality implies that 
$\epsilon_1 +\epsilon_2 \geq 1-\frac{1}{\sqrt{2}}$ 
must be satisfied so that both 
$\max_{z\in \Omega_Z} \left|
O_d \left(z, {\cal W}_{\epsilon_1}(M_1,Z)\right)\right|=1
$
and 
$\max_{x\in \Omega_X}\left| 
O_d \left(
x,{\cal W}_{\epsilon_2}(M_2,X)
\right)
\right|=1
$
hold. 
Choose $\epsilon_1$ and $\epsilon_2$ so as to 
satisfy 
\begin{eqnarray*}
1-\epsilon_1 &=&\min_{z\in \Omega_Z}
\mbox{tr}\left(
|z\rangle \langle z| M_{1z}
\right)
\\
1-\epsilon_2 &=&\min_{x\in \Omega_X}
\mbox{tr}\left(
|x\rangle \langle x| M_{2x}
\right), 
\end{eqnarray*}
where $M_1$ and $M_2$ are marginal POVMs of $M$. 
Because $\mbox{tr}\left(|z\rangle \langle z| M_{1z} \right)
\geq 1-\epsilon_1$ is satisfied 
for each $z\in \Omega_Z$, 
$\max_{z\in \Omega_Z} \left|
O_d \left(z, {\cal W}_{\epsilon_1}(M_1,Z)\right)\right|=1
$ holds. 
Similarly, $\max_{x\in \Omega_X}\left| 
O_d \left(
x,{\cal W}_{\epsilon_2}(M_2,X)
\right)
\right|=1
$ follows. 
Thus it holds 
that 
\begin{eqnarray*}
\min_{z\in \Omega_Z}
\mbox{tr}\left(
|z\rangle\langle z| M_{1z}
\right)+
\min_{x\in \Omega_X}
\mbox{tr}\left(
|x\rangle \langle x| M_{2x}
\right)
\leq 1+\frac{1}{\sqrt{2}}.
\end{eqnarray*}
\end{example} 
\begin{example}(Position and momentum on a torus) 
\par
We consider a quantum particle confined in a torus ${\bf T^2}
=[0,2\pi) \times [0,2\pi)$ with $\hbar =\frac{2\pi}{N}$ for 
some integer $N\geq 2$. 
The system is described 
by a Weyl algebra (or rotation algebra) which is 
generated by $u$ and $v$ satisfying 
$u^*u=v^*v={\bf 1}$ and $uv=e^{i\hbar}vu$. 
The algebra can be represented on an $N$-dimensional Hilbert space.  
${\cal H}\simeq {\bf C}^N$ has an orthonormalized basis $\{|n\rangle\}_{n=0,1,\ldots,N}$.  
$u$ and $v$ act on ${\cal H}$ as
 $u|n\rangle=e^{i\frac{2\pi n}{N}}|n\rangle$ and 
$v|n\rangle =|n+1\rangle$. 
A PVM corresponding to position is 
$Q=\{|n\rangle \langle n|\}_{n\in \Omega_Q}$, 
and a PVM corresponding to momentum is 
$P=\{|\overline{k}\rangle \langle \overline{k}|\}_{k\in 
\Omega_P}$, where 
$\Omega_Q=\Omega_P=\{0,1,\ldots,N-1\}$ 
and $|\overline{k}\rangle =
\frac{1}{\sqrt{N}}\sum_{n=0}^{N-1} e^{-i\frac{2\pi kn}{N}}|n\rangle$. 
Their noncommutativity is expressed by 
$\mbox{tr}(|n\rangle \langle n|\overline{k}\rangle 
\langle \overline{k}|)=\frac{1}{N}$ for each $n\in \Omega_Q$ 
and $k\in \Omega_P$. 
A metric on $\Omega_Q$ is introduced by 
$d_Q (m,n):=\min\{\frac{2\pi}{N}(m-n+Nt)|t\in {\bf Z}, 
m-n+Nt \geq 0\}$, 
and a metric on $\Omega_P$ is introduced by 
$d_P (k,l):=\min\{\frac{2\pi}{N}(k-l+Ns)|s\in {\bf Z}, 
k-l+Ns \geq 0\}$. 
Because $|O_{d_Q}(m,w_1)|\leq \frac{w_1N}{2\pi}+1$ 
and $|O_{d_P}(n,w_2)|\leq \frac{w_2N}{2\pi}+1$ hold 
for any $m\in \Omega_Q$, $n\in \Omega_P$, and $w_1,w_2\geq 0$, 
it holds for any $\rho \in {\cal S}({\cal H})$, and $\epsilon_1,\epsilon_2
\in [0,1]$ with $\epsilon_1+\epsilon_2\leq 1$ that 
\begin{eqnarray*}
\left(W_{\epsilon_1}(\rho^Q)+\frac{1}{N}\right)\cdot
\left( W_{\epsilon_2}(\rho^P)+\frac{1}{N}\right)
\geq 2\pi 
\hbar
(1-\epsilon_1-\epsilon_2)^2. 
\end{eqnarray*}
In addition, it holds for any POVM $M=\{M_{nk}\}_{(n,k)\in \Omega_Q
\times \Omega_P}$ that 
\begin{eqnarray*}
\left({\cal W}_{\epsilon_1}(M_1,Q)+\frac{1}{N}\right)\cdot
\left( {\cal W}_{\epsilon_2}(M_2,P)+\frac{1}{N}\right)
\geq 2\pi
\hbar 
(1-\epsilon_1-\epsilon_2)^2. 
\end{eqnarray*}
\end{example}
\begin{example}(N-qubit) 
\par
Consider an $N$-qubit system described by a Hilbert space 
${\cal H}={\bf C}^2\otimes {\bf C}^2\otimes \cdots \otimes {\bf C}^2$ 
($N$ times). We write a standard basis 
$|z\rangle =|z_1\rangle \otimes |z_2
\rangle \otimes \cdots \otimes |z_N\rangle$ 
for $z=z_1z_2\cdots z_N\in \{0,1\}^N$.
Its conjugate basis is written as
$|\overline{x}\rangle =|
\overline{x_1}\rangle \otimes |\overline{x_2}\rangle 
\otimes \cdots \otimes |\overline{x_N}\rangle$, 
where $|\overline{0}\rangle:=\frac{1}{\sqrt{2}}
(|0\rangle +|1\rangle)$ 
and $|\overline{1}\rangle:=
\frac{1}{\sqrt{2}}
(|0\rangle -|1\rangle)$. 
We consider a pair of PVMs $X=\{X_x\}_{x\in \Omega_X}=\{|\overline{x}\rangle
\langle \overline{x}|\}_{x\in \Omega_X}$  and $Z=\{Z_z\}_{z\in \Omega_Z}
=\{|z\rangle \langle z|\}_{z\in \Omega_Z}$, both of which are 
defined on a metric space $\Omega_X=\Omega_Z=\{0,1\}^N$.  
While $\{0,1\}^N$ is
equipped with a natural distance called Hamming distance 
$d_H$, we employ its rescaled version $d$ defined as
$d(x_1,x_2):=\frac{d_H(x_1,x_2)}{2N}$ for $x_1,x_2\in \{0,1\}^N$. 
The noncommutativity between $X$ and $Z$ is quantified as 
$\mbox{tr}(X_x Z_z)=\frac{1}{2^N}$. 
For any $w\in [0,\frac{1}{4}]$ and $x\in \{0,1\}^N$, 
it holds that $|O(x, w)|=|\{y\in \{0,1\}^N| d_H(x,y) \leq wN\}| 
\leq 2^{Nh(w)}$, where $h(w):=-w\log w
-(1-w)\log (1-w)$. 
Thus it holds for any $\epsilon_1,\epsilon_2\in [0,1)$ 
with $\epsilon_1+\epsilon_2<1$ and 
for any state $\rho$ that 
\begin{eqnarray*}
h(W_{\epsilon_1}(\rho^Z))
+h(W_{\epsilon_2}(\rho^X))\geq 1+\frac{2}{N}\log (1-\epsilon_1
-\epsilon_2).
\end{eqnarray*}
In addition it holds for any POVM $M$ on $\Omega_Z\times \Omega_X$ that 
\begin{eqnarray*}
h({\cal W}_{\epsilon_1}(M_1,Z))
+h({\cal W}_{\epsilon_2}(M_2,X))\geq 1+\frac{2}{N}\log (1-\epsilon_1
-\epsilon_2).
\end{eqnarray*}
\end{example}
\section{Other measures for joint measurability}\label{others}
\subsection{Werner's geometric measure}
In \cite{Werner}, Werner discussed the uncertainty relation for joint measurement 
of position and momentum by employing Monge distance on a space of 
probability measures. While he defined a distance between a pair of 
POVMs on a general metric space, we treat only
 a finite metric space in this paper. 
His key idea is to consider the Lipshitz ball of a space of functions 
in order to take into consideration the metric structure 
of the space $(\Omega,d)$. 
The Lipshitz ball is defined by 
$\Lambda:= \{f:\Omega \to {\bf R}|\forall x, y \in \Omega,  |f(x)-f(y)|\leq d(x,y)\}$.  
He defined a distance between 
a pair of POVMs $F_1$ and $F_2$ on a metric space $(\Omega,d)$ by 
\begin{eqnarray*}
D_W(F_1,F_2)
&:=&\sup_{\rho\in {\cal S}({\cal H})}
\sup_{f\in\Lambda}|\mbox{tr}(\rho (F_1(f)-F_2(f))|\\
&=&\sup_{f\in \Lambda}\Vert F_1(f)-F_2(f)\Vert, 
\end{eqnarray*} 
where $F(f):=\sum_{x\in \Omega}F_x f(x)$ for a POVM $F$ and 
a function $f:\Omega\to {\bf R}$. 
\par
Busch and Pearson \cite{Buscherrorbar} proved the following relationship between 
the error bar width for finite $\epsilon$ and the geometric distance. 
(While their statement is restricted on observables on ${\bf R}$, 
the proof holds for observables on a general metric space $(\Omega,d)$. )
\begin{prop}(\cite{Buscherrorbar})
For any POVM $E_1$ on a (finite) metric space $(\Omega,d)$, 
a PVM 
$E$ on $(\Omega,d)$ and $\epsilon \in (0,1]$,  it holds that 
\begin{eqnarray*}
{\cal W}_{\epsilon} (E_1,E)\leq \frac{2}{\epsilon}D_W (E_1,E). 
\end{eqnarray*}
\end{prop} 
Combining this proposition with our theorem, we obtain the following. 
\begin{theorem}
Let ${\cal H}$ be an $N$-dimensional Hilbert space. 
Suppose that we have a pair of PVMs $A=\{A_x\}_{x\in \Omega_A}$ 
and $B=\{B_y\}_{y \in \Omega_B}$, where $\Omega_A$ (resp. $\Omega_B$) 
is equipped with a metric $d_A$ (resp. $d_B$). 
We consider a POVM $M=\{M_{xy}\}_{(x,y)\in \Omega_A\times \Omega_B}$ 
to perform an approximate joint measurement of $A$ and $B$.
$M_1$ and $M_2$ denote marginal POVMs obtained from $M$.  
For any $\epsilon_1, \epsilon_2 \in (0,1)$ with $\epsilon_1+\epsilon_2<1$, 
there exists $\rho\in {\cal S}({\cal H})$ satisfying
\begin{eqnarray*}
D_W(M_1,A)&\geq& \frac{\epsilon_1}{2}W_{\epsilon_1+\epsilon_2}
(\rho^{A}) \\
D_W(M_2,B)&\geq& \frac{\epsilon_2}{2}W_{\epsilon_1+\epsilon_2}
(\rho^{B}).
\end{eqnarray*}
\end{theorem}
\subsection{$l_{\infty}$ distance}
In \cite{MiyaHeisen}, 
Miyadera and Imai discussed the uncertainty relation for 
a pair of general discrete POVMs by employing $l_{\infty}$ distance 
on a space of probability distributions. 
For POVMs $E^1=\{E^1_x\}_{x\in \Omega}$ and $E^2=\{E^2_x\}_{x\in \Omega}$, 
their $l_{\infty}$ distance is defined by 
\begin{eqnarray*}
D_{\infty}(E^1,E^2)&:=&\sup_{\rho} \max_{x\in \Omega}
|\mbox{tr}(\rho(E^1_x -E^2_x))|
\\
&=&\max_{x\in \Omega} \Vert E^1_x -E^2_x\Vert. 
\end{eqnarray*}
As this measure does not consider a metric structure, 
it is not related to the overall width directly. 
Instead, we need the following quantity characterizing 
joint localizability. 
\par
For any probability distribution on a finite set $\Omega$, 
we define its minimum localization error by 
\begin{eqnarray*}
LE(p):=1- \max_{x\in \Omega} p(x).
\end{eqnarray*}
This quantity is vanishing if and only if there exists 
$x\in \Omega$ such that $p(x)=1$. 
For POVM $A=\{A_x\}_{x\in \Omega}$ 
and state $\rho\in {\cal S}({\cal H})$,  
$LE(\rho^A)$ characterizes the localization property 
of $\rho^A$. We have the following theorem. 
\begin{theorem}\label{thimaichi}
Let ${\cal H}$ be an $N$-dimensional Hilbert space. 
Suppose that we have a pair of PVMs $A=\{A_x\}_{x\in \Omega_A}$ 
and $B=\{B_y\}_{y \in \Omega_B}$. 
We consider a POVM $M=\{M_{xy}\}_{(x,y)\in \Omega_A\times \Omega_B}$ 
to perform an approximate joint measurement of $A$ and $B$.
$M_1$ and $M_2$ denote marginal POVMs obtained from $M$.  
\par
There exists $\rho \in {\cal S}({\cal H})$ satisfying 
\begin{eqnarray*}
D_{\infty}(M_1,A)+D_{\infty}(M_2,B)\geq LE(\rho^A)+LE(\rho^B). 
\end{eqnarray*}
\end{theorem}
\begin{proof}
We sketch the proof as it employs a technique similar to the proof of Theorem \ref{th1}.  
We use the quantities introduced there. 
Consider the maximally entangled state $\Phi_0
\in {\cal H}_{aux}\otimes {\cal H}$ and 
the remote state preparation by using $\{\gamma(A_x)\}_{x\in \Omega_A}$ 
and $\{\gamma(B_y)\}_{y\in\Omega_B}$ as in the proof of Theorem \ref{th1}.
We measure $M$ on ${\cal H}$. 
From the definition of $l_{\infty}$ distance, 
there exist $x\in \Omega_A$ and $y\in \Omega_B$ such that 
\begin{eqnarray*}
1-\sum_{y'\in\Omega_B} p_{M|A}(x,y'|x)&\leq& D_{\infty}(M_1,A)
\\
1-\sum_{x'\in\Omega_A} p_{M|B}(x',y|y)&\leq& D_{\infty}(M_2,B) 
\end{eqnarray*}
hold. 
It follows that 
\begin{eqnarray*}
1-\sum_{x\in \Omega_A} \sum_{y\in\Omega_B} p_{M|A}(x,y|x)
p_A(x)&\leq& D_{\infty}(M_1,A)
\\
1-\sum_{x\in\Omega_A}\sum_{y\in \Omega_B} p_{M|B}(x,y|y)p_B(y)
&\leq& D_{\infty}(M_2,B). 
\end{eqnarray*}
Thanks to the Bayes theorem, these inequalities imply 
\begin{eqnarray*}
1- \sum_{x,y}p_{A|M}(x|x,y)p_M(x,y)&\leq& D_{\infty}(M_1,A)\\
1-\sum_{x,y}p_{B|M}(y|x,y)p_M(x,y)&\leq& D_{\infty}(M_2,B). 
\end{eqnarray*}
Thus there exists $(x,y)\in \Omega_A\times \Omega_B$ such that 
\begin{eqnarray*}
(1-p_{A|M}(x|x,y))+(1-p_{B|M}(y|x,y))\leq D_{\infty}(M_1,A)+D_{\infty}(M_2,B). 
\end{eqnarray*}
As $p_{A|M}(x|x,y)=\mbox{tr}(\gamma^* (\Theta_{xy})A_x)$ 
and $p_{B|M}(y|x,y)=\mbox{tr}(\gamma^*(\Theta_{xy})B_y)$ hold, 
we obtain 
\begin{eqnarray*}
LE(\gamma^*(\Theta_{xy})^A)+LE(\gamma^*(\Theta_{xy})^B)
\leq D_{\infty}(M_1,A)+D_{\infty}(M_2,B). 
\end{eqnarray*}
This completes the proof. 
\end{proof}
As the right-hand side of the inequality in Theorem \ref{thimaichi} can be 
bounded by the Landau-Pollak uncertainty relation, 
we can obtain a bound for the $l_{\infty}$ distance. 
\begin{corollary}\label{corD}
Let ${\cal H}$ be an $N$-dimensional Hilbert space. 
Suppose that we have a pair of PVMs $A=\{A_x\}_{x\in \Omega_A}$ 
and $B=\{B_y\}_{y \in \Omega_B}$. 
We consider a POVM $M=\{M_{xy}\}_{(x,y)\in \Omega_A\times \Omega_B}$ 
to perform an approximate joint measurement of $A$ and $B$.
$M_1$ and $M_2$ denote marginal POVMs obtained from $M$.  
It holds that 
\begin{eqnarray*}
D_{\infty}(M_1,A)+D_{\infty}(M_2,B)\geq 1-\max_{(x,y)\in \Omega_A \times 
\Omega_B} \Vert A_x B_y\Vert. 
\end{eqnarray*}
\end{corollary}
It is interesting to compare this corollary with 
a bound obtained in \cite{MiyaHeisen}: 
\begin{eqnarray}
&&2D_{\infty}(M_1,A)D_{\infty}(M_2,B)
+D_{\infty}(M_1,A)+D_{\infty}(M_2,B)
+4 D_{\infty}(M_1,A)^{1/2} D_{\infty}(M_2,B)^{1/2}
\geq \max_{x,y}\Vert [A_x,B_y]\Vert. 
\label{miyah}
\end{eqnarray}
A simple example treating a qubit illustrates 
that Corollary \ref{corD} is better in some region 
and worse in other region than (\ref{miyah}). 
We will investigate their comparison in detail elsewhere. 
\section{discussion}
In this paper, we investigated the relationships between 
two aspects of quantum uncertainty: joint localizability 
and joint measurability. 
To characterize these properties, 
Busch and Pearson introduced the overall width and the error bar width. 
In Theorem \ref{th1}, we showed a simple relationship between them. 
Our theorem implies that if there is a bound on 
joint localizability, it is possible to obtain a 
similar bound on joint measurability. 
While in heuristic arguments, the bound for joint measurement is 
often substituted without any rationale by that for joint localizability,  
which is easier to derive, 
our result provides a reasonable basis for this 
substitution.  
In addition, as by-products, we proved trade-off inequalities representing 
the limitations to joint localizability and measurability for 
a pair of general PVMs in a finite-dimensional system.
In contrast to Busch and Pearson's derivation, 
we did not require any concrete form of 
$M$. 
\par
Although our result is applicable to a pair of general 
PVMs in finite-dimensional systems, it is impossible to 
treat a single particle on ${\bf R}$, which 
needs an infinite-dimensional Hilbert space. 
Because our proof of Theorem \ref{th1} mainly depends on 
the properties of vector $\Phi_0$ whose form is 
meaningless in an infinite dimensional Hilbert space, 
it is impossible to extend our theorem directly 
to the infinite-dimensional case. 
Nevertheless, as Keyl, 
Schlingemann and  Werner discussed \cite{Wernerinfinite}, 
some observable algebras yield a maximally entangled state
as a linear functional. In the GNS representation of 
such a state, one can define an antiunitary operator $J$ 
that works for the remote state 
preparation. 
In such a case, even if $\Omega_A$ and $\Omega_B$ are 
continuous, they can be treated almost similarly to the  
finite-dimensional case if $\Omega_A$ and $\Omega_B$ 
are compact. This, however, is not the case for a single particle moving 
on ${\bf R}$. Normally, its Hilbert space is described by 
${\cal H}=L^2({\bf R})$. 
One can define 
an EPR state as a linear functional over the Weyl algebra 
for two particles. 
Although the state cannot be represented 
as a density operator in the original Hilbert space, 
because of the Hahn-Banach theorem, the linear functional can be extended 
to ${\bf B}({\cal H}_{aux})\otimes {\bf B}({\cal H})$. 
Another problem arises from the non-compactness of 
${\bf R}$: because the EPR state is singular, 
for any compact set $\Delta \subset {\bf R}$, 
the expectation values of $Q(\Delta)$ and $P(\Delta)$ 
vanish. Therefore, 
depending on the probability, our argument cannot be 
justified. In fact, the expectation value of $Q(\Delta)$ for $\Delta\subset {\bf R}$ 
does not satisfy $\sigma$-additivity. 
To treat these issues, more 
mathematical detail is required. 
\par
Our theorem shows that if a pair of PVMs 
is not jointly localizable, the pair is not jointly measurable. 
In general, the opposite does not hold. 
Let us consider a three-dimensional Hilbert space 
${\bf C}^3={\bf C}\oplus {\bf C}^2$. We define a pair of PVMs $A$ and $B$ as follows. 
$A=\{A_0,A_1,A_2\}$ is defined by $A_0={\bf 1}_{\bf C}\oplus {\bf 0}$, 
$A_1=|0\rangle\langle 0|$ and $A_2=|1\rangle \langle 1|$. 
$B=\{B_0,B_1,B_2\}$ is defined by $B_0={\bf 1}_{\bf C}\oplus {\bf 0}$, 
$B_1=| \overline{0}\rangle \langle \overline{0}|$ and 
$B_2=|\overline{1}\rangle \langle \overline{1}|$, where 
$|\overline{0} \rangle :=\frac{1}{\sqrt{2}}(|0\rangle +|1\rangle)$ 
and $|\overline{1} \rangle :=\frac{1}{\sqrt{2}}(|0\rangle -|1\rangle )$. 
These observables are jointly localizable but are not jointly measurable. 
However, because this example is somewhat artificial, it would be interesting 
to pursue a suitable definition of joint localizability to overcome this unconformity. 
In addition, from a foundational viewpoint, 
it would be interesting to examine whether the relationship 
between joint localizability and measurability is characteristic 
in quantum theory. For instance, in a broad framework called quantum logic 
that enables us to treat general theories, we can construct 
a theory that yields joint measurability but does not have 
joint localizability. 
In fact, a theory described by the Fano plane 
\cite{Dvurecenskij}
has only one unique state 
that is not localized with respect to any observables. 
Any observables in this theory, however, are trivially jointly measurable. 
This example is rather artificial and could be avoided 
if one introduces additional natural conditions. 
We hope to discuss such problems elsewhere. 
\\
{\bf Acknowledgments:} 
I would like to thank an anonymous referee for helpful comments. 


\begin{thebibliography}{9}
\bibitem{metric}
That is, $d:\Omega \times \Omega \to [0,\infty)$ satisfies 
(i) symmetry: $d(x,y)=d(y,x)$ for all $x,y \in \Omega$, 
(ii) positive definiteness: $d(x,y)=0$ if and only if $x=y$, 
and (iii) triangular inequality: 
$d(x,y)+d(y,z)\geq d(x,z)$ for all $x,y,z\in \Omega$.
\end{thebibliography}
\end{document}